\documentclass[journal,onecolumn]{IEEEtran}

\usepackage{amsmath,amsfonts, amssymb, amsthm}
\usepackage{algorithmic}
\usepackage{algorithm}
\usepackage{array}
\usepackage[caption=false,font=normalsize,labelfont=sf,textfont=sf]{subfig}
\usepackage{textcomp}
\usepackage{stfloats}
\usepackage{url}
\usepackage{hyperref}
\usepackage{verbatim}
\usepackage{graphicx}
\usepackage{cite}
\newtheorem{theorem}{Theorem}
\newtheorem{lemma}{Lemma}
\newtheorem{corollary}{Corollary}
\newtheorem{definition}{Definition}

\newtheorem{remark}{Remark}
\newtheorem{example}{Example}

\usepackage{tikz, graphicx}
\usepackage{tikz-3dplot}

\begin{document}

\title{The Service Rate Region of Hamming Codes}

\author {Priyanka~Choudhary 
       and~Maheshanand~Bhaintwal\\ Department of Mathematics, Indian Institute of Technology Roorkee, Roorkee 247667, India
       
       \thanks{Priyanka Choudhary is a Ph.D. scholar in the Department of Mathematics, Indian Institute of Technology Roorkee, Roorkee 247667, India (e-mail: \href{mailto:priyanka_c@ma.iitr.ac.in}{priyanka\_c@ma.iitr.ac.in}).}
  \thanks{Maheshanad Bhaintwal is with the Department of Mathematics, Indian Institute of Technology Roorkee, Roorkee 247667, India (e-mail: \href{mailto:maheshanand@ma.iitr.ac.in}{maheshanand@ma.iitr.ac.in}).}}



\maketitle

\begin{abstract}
The service rate region of a coded distributed storage system is the set of all achievable data access requests under the capacity constraints. This paper investigates the service rate regions of systematic Hamming codes using hypergraph theory and derives bounds for the maximal achievable service rate of individual data objects. We establish upper bounds on the sum of service rates of data symbols indexed by a subset of systematic nodes in a systematic binary Hamming code, and explore the achievability of these bounds. Additionally, for non-systematic binary Hamming codes, we conclude that the aggregate service rate is limited by the number of columns of odd weight in the associated generator matrix.
\end{abstract}

\begin{IEEEkeywords}
Service rate region, hypergraph, vertex cover, Hamming codes, recovery sets.
\end{IEEEkeywords}

\IEEEpeerreviewmaketitle

\section{Introduction}

\IEEEPARstart{D}{istributed} data storage systems have become integral to modern data management, ensuring reliability, scalability, and availability in diverse applications ranging from cloud computing to large-scale storage. To maintain accessibility and safeguard against potential failures, storage systems employ redundancy, either through data replication or coding techniques. While replication has traditionally been a popular method, it is no longer a practically feasible method for exponentially increasing data in modern times. As a result, erasure coding has emerged as a more efficient alternative. Unlike replication, erasure coding reduces storage overhead while maintaining reliability and improving data retrieval efficiency through local recovery, as detailed in works like \cite{dimakis2010network, dimakis2011survey, gopalan2012locality}, and \cite{ huang2013pyramid}. 
Erasure coding also effectively enhances availability by reducing download latency for retrieving entire data, as highlighted in several research papers, see e.g., \cite{joshi2012coding, liang2013fast, gardner2015reducing}. Among the erasure coding techniques, binary Hamming codes, renowned for their efficient error-correction properties and low overhead, present an interesting area to explore.

Whenever new erasure coding schemes are developed and presented, there is often a question regarding which one is the most effective for a given application. To assess the performance of different coded storage systems, various factors are compared to ensure the selection of a scheme that best aligns with the user's requirements. One critical metric for evaluating the performance of erasure-coded storage systems is the ``\emph{service rate region (SRR)}''. This region represents the set of download or access request rates that the system can fulfill without exceeding its service capacity. This concept was initially introduced by Noori et al. \cite{noori2016storage}. Later, Akta{\c{s}} et al. \cite{aktacs2017service} extended the framework of storage allocation by connecting it to the service rate region. Characterizing the SRR provides a clear understanding of the overall rate of requests that a system can handle. In simpler terms, the SRR metric serves as a measure of assessing the operational efficiency of distributed coded systems under the conditions of parallel access, while considering the constraints imposed by the limited service capacity of the network's nodes.

This metric is particularly important in scenarios where multiple users access shared content concurrently, often with uneven access patterns. It also highlights the system's resilience to varying request patterns (skewed demand distribution), where certain objects are accessed more frequently than others.

There are two primary problems associated with SRR: \begin{enumerate}
    \item Determining the SRR for a given coding scheme.
    \item  For a given region $R$, determining the coding scheme whose SRR contains $R$ with the minimum number of storage nodes.
\end{enumerate}
Research on SRR aims to find solutions to the problems mentioned above.   
Finding the service rate region of a scheme involves determining the optimal allocation of incoming data download requests to the recovery sets. This process effectively transforms the task of maximizing the service rate region into an optimization problem.

In \cite{kazemi2020geometric}, upper bounds for the SRR of binary simplex codes and first-order Reed-Muller codes were determined through a projective geometric approach.  Ly et al.\cite{ly2025RM} generalized their approach to analyze the service rate region of higher-order RM codes, providing better inner and outer bounds.

It was realized in \cite{kazemi2020combinatorial} that the problem of finding a valid allocation for SRR was related to identifying fractional matching in hypergraphs, and is similar to batch codes. Alfarano et al.\cite{alfarano2024service} explored the polytope structure of the service rate region and proved that every rational $k$-tuple of the service rate region has a valid rational allocation. In \cite{alfarano2022dual}, the authors introduced two dual distance outer bounds for the SRR by using the structural properties of the generator matrix. 

 The second significant challenge is identifying the minimum number of storage nodes $n_R$ and the best coding scheme for the given region $R$. This aspect has been analyzed in \cite{kazemi2021efficient}. The development of a coding scheme for a fixed number of servers aimed at maximizing the SRR was investigated in \cite{aktacs2017service} for the case when $k=2$, and subsequently in \cite{anderson2018service} for $k=3$. In 2023, Alfarano et al.\cite{alfarano2023service} investigated the SRR of locally recoverable codes (LRCs) and batch codes. Kiliç et al.  \cite{kilicc2024parameters} have provided both upper and lower bounds for the minimum number of servers $n_R$ corresponding to a fixed field size. Moreover, for a fixed number of servers, they have obtained bounds on the field size as well. 
 
 This research extends the idea of recent developments in the characterization of SRR for MDS-coded distributed systems. Ly and Soljanin\cite{ly2025service} established a graph-theoretic framework that connects SRR to fractional matchings (or fractional vertex cover) in quasi-uniform hypergraphs, demonstrating that the SRR polytope expands with systematic columns in generator matrices. There are still many questions that remain unanswered regarding the SRR of various coding schemes, particularly when the generator matrix is non-systematic.
 
 Compared to other coding schemes, characterizing the service rate region of binary Hamming codes using the hypergraph framework is not overly complex due to their structural properties and unique form of parity-check matrices. Recently, Ly and Soljanin \cite{ly2025m} explored the connection between combinatorial design theory and SRR. They established that the maximum achievable demand of any data object in a linear code is attained when the minimum weight codewords of the dual code form a $2$-design. Through this analysis, they proved that, for systematic binary Hamming codes, the maximum achievable demand of a data object is consistently $3$, irrespective of the code's parameters. This result is also obtained in the present paper through a different approach by analyzing the recovery system structure of Hamming codes. 
 
 We explore the service rate region of binary Hamming codes thoroughly using the existing approach of hypergraphs \cite{aktacs2021service} \cite{ly2025service} in the context of SRR and attempt to find more accurate bounds to characterize their service rate region. 

 This paper is organized as follows. Section \ref{ii} includes some preliminaries and the fundamental concepts of the service rate region. It also examines the challenge of identifying the SRR from various angles, including combinatorial and geometric perspectives. In Section \ref{iii}, we present the principal findings of this paper, concentrating on the structure of the recovery $G$ system and the maximal achievable service rate of a data object for systematic $q$-ary Hamming codes, denoted by Ham($r, q$). We also establish an upper bound for the cumulative service rate of non-systematic Ham($r, 2$). Furthermore, we outline an upper limit on the achievable demand of data files corresponding to a subset $\mathcal{I}$ of systematic nodes of systematic Ham($r, 2$) and demonstrate that this limit is indeed attainable. Section \ref{iv} concludes the paper and proposes some potential directions for future research.

\section{Preliminaries}\label{ii}
 This section introduces the notations, definitions, and prior results that form the foundation of the service rate region of coded schemes in data storage systems. Unless otherwise stated, all notations and terminologies follow standard conventions.
 
 Let $\mathbb{F}_q$ be a finite field of size $q$. A linear code $C$ of length $n$ and dimension $k$ over $\mathbb{F}_q$ is a $k$-dimensional subspace of $\mathbb{F}_q^n$, and referred to as an $[n, k]$-code. Given an information vector with $k$ symbols, an $n$-symbol codeword can be generated using a valid encoding process. The set of coordinate positions with non-zero entries in a codeword $c\in C$ is known as the support of $c$, i.e., for any $c\in C$, the support of $c$, denoted by $\text{Supp}(c)$ is \[\text{Supp}(c)=\{i\mid c_i\neq 0,  1\leq i\leq n\}.\]

  The minimum distance of a code $C$ is defined as
$d(C) = \min
\{d(x, y): x, y\in C, x\ne y\},$ where $d(x, y)=|\{i : x_i\neq y_i\}|.$
For a linear code $C$, its orthogonal complement is referred to as the dual code of $C$ and denoted by $C^{\perp}$. Thus,  \[C^{\perp} =\{x \in \mathbb{F}_q^n \mid x\cdot c=0, \text{for all } c \in C\},\]
  where $x\cdot c$ denotes the usual inner product of $x$ and $c$.
     
     One can also characterize an $[n, k]$ linear code $C$ by a $k\times n$ generator matrix $G$ of rank $k$ over $\mathbb{F}_q$. The rows of $G$ form an $\mathbb{F}_q$ basis for $C$. We have $x=(x_1, x_2, \cdots, x_k)\in \mathbb{F}_q^k$ as information (data) objects and $x\cdot G \in \mathbb{F}_q^n$ is the corresponding codeword where the $k$ data objects are stored. Then, $C$ can be described as the set $C=\{x\cdot G \mid x\in \mathbb{F}_q^k\}$. An $(n-k)\times n$ matrix $H$ of rank $n-k$ over  $\mathbb{F}_q$ is called a parity-check matrix of $C$ if $G\cdot H^\top=0$. A parity-check matrix $H$ of $C$ is also a generator matrix for the dual code of $C$.

We consider a coded distributed system consisting of $n$ servers to store $k$ information(data) files represented by elements of $\mathbb{F}_q$. Let $C$ be the corresponding code over $\mathbb{F}_q$ and let $G$ denote a generator matrix of $C$. One server stores exactly one element of $\mathbb{F}_q$ and has two major functions: one is to store the data, and the other is to process the incoming data access requests from the system's users. Each server possesses a service capacity denoted by $\mu$, which indicates that it is capable of processing $\mu$ requests per unit of time. In this storage system, each server stores a linear combination of the $k$ information files, and each column of the matrix $G$ uniquely represents a particular server. A generator matrix $G$ is said to be systematic if there are $k$ columns in the matrix which are the ordered standard basis vectors $e_1, e_2, \cdots, e_k$ of $\mathbb{R}^k$. The columns of $G$ that are non-zero scalar multiples of the standard basic vectors are known as systematic servers. 
For an $[n, k]$ linear code $C$ over $\mathbb{F}_q$, the standard form of a systematic generator matrix is $G=[I_k\mid A]$, and it is well known that $H=[-A^\top\mid I_{n-k}]$ is the corresponding parity-check matrix for $C$. If another systematic matrix $G'$ of $C$ is obtained through permuting the columns of $G$, then the corresponding parity-check matrix $H'$ is obtained by applying the same column permutation to $H$. This convention will be followed throughout this paper. 

As the information files are saved across the $n$ servers in coded or systematic servers, we can recover these files from multiple subsets $I\subset[n]=\{1, 2, ..., n\}$. This leads to the following definition.

\begin{definition}
    A set $R\subset[n]$ is called a recovery set for the $i$-th data object if the standard basis vector $e_i \in \mathbb{F}_q^k$ belongs to the $\mathbb{F}_q$-linear span of the column vectors of $G$ indexed by the elements in $R$.
\end{definition}

The set of all recovery sets for the $i$-th information file is denoted by $\mathcal{R}_i^{all}$. 
We define the set $\mathcal{R}(G)=\{ \mathcal{R}_1,  \mathcal{R}_2, ...,  \mathcal{R}_k\}$ as a \textit{recovery $G$ system}\cite{aktacs2021service}, where each component $\mathcal{R}_i \subseteq  \mathcal{R}_i^{all}(G)$. The recovery system containing all the recovery sets corresponding to the generator matrix $G$ is denoted by $\mathcal{R}^{all}(G)$.

 An element $R$ of $\mathcal{R}_i^{all}(G)$ is called a minimum recovery set for the $i$-th data object if there does not exist $R' \in \mathcal{R}_i^{all}(G)$ such that $R' \subsetneq R$. A minimum recovery $G$ system is \[\mathcal{R}^{min}(G)=\{\mathcal{R}^{min}_i  \mid 1 \leq i \leq k\},\] where $\mathcal{R}^{min}_i $ contains all the minimum recovery sets of the $i$-th object. 
 
 \begin{example}
     A systematic code $C$ has a systematic generator matrix. Therefore, a minimum recovery set for a data object consists only of the corresponding systematic column of the generator matrix.
 \end{example}

Consider a linear code $C$ characterized by a generator matrix $G$ that can be transformed into a systematic form through a permutation of its columns. In this case, the parity-check matrix helps in completely characterizing the recovery system using the codewords of the corresponding dual code. Consequently, we have the following results.

\begin{lemma}\cite{alfarano2022dual}\label{Rset} Let $G$ be a systematic generator matrix of a linear code $C$. A subset $R\subset [n]$ is a recovery set for the $i$-th coordinate if and only if one of the following conditions is satisfied:
\begin{enumerate}
    \item The set $R$ contains a systematic server of the $i$-th coordinate. 
    \item  There exists a codeword $c'\in C^\perp$ such that the Supp$(c')\subseteq R\cup \{i\}$ and $i\in $ Supp$(c')$.
\end{enumerate}
\end{lemma}
 \begin{lemma}\label{dperp}
    Let $G$ be a systematic generator matrix of a linear code $C$ or column permutation equivalent to a systematic generator matrix. If $R$ is a recovery set for the $i$-th coordinate in $C$, then either $|R|=1$ or $|R|\geq d^\perp-1,$ where $d^\perp$ is the minimum distance of the dual code of $C$.
 \end{lemma}
\begin{proof}
    When $R=\{i\}$, it is obvious that $R$ is a recovery set for the $i$-th coordinate. Let us consider $R$ a recovery set such that $R\neq\{i\}$. 
    Using Lemma \ref{Rset}, we know that there exists a codeword $c'\in C^\perp$ such that the Supp$(c')\subseteq R\cup \{i\}.$
    Therefore we have, \[d^\perp \leq |\text{Supp}(c')|\leq |R|+1.\]\end{proof}
    
    Continuing in this direction, we present the following result that helps in determining the number of recovery sets of a data file using the codewords of the dual code.

    \begin{lemma}\label{ithdata}\cite{ling2004coding}
        Consider an $[n, k, d]$ linear code $C$ over $\mathbb{F}_q$. There are exactly $1/q$ of the codewords of $C$ which have $\alpha\in\mathbb{F}_q$ at any given coordinate position $i$, or zero in all the codewords at the $i$-th coordinate position, $1\leq i\leq n$.
    \end{lemma}
 For a code $C$ with a systematic generator matrix $G$, let $S\subset[n]$ such that $|S|=k$ and all the columns of $G$ indexed by the elements of $S$ are the distinct $k$ systematic columns. From Lemma \ref{ithdata}, for each $i\in S$, there are $q^{n-k-1}$ codewords in $C^\perp$ that have $1$ in coordinate position $i$. Therefore, the number of recovery sets $|\mathcal{R}^{\min}_i|$ of the data object at the $i$-th place satisfies $|\mathcal{R}^{\min}_i|\geq q^{n-k-1}$ for all $i\in S$. This observation will lead to important results in Section~\ref{iii}.

In the following subsection, we provide a formal definition of the service rate region associated with a given recovery $G$ system, as well as analyze the bounds previously established for this region.
\subsection{Service Rate Region}

Let $\lambda_i\in\mathbb{R}_{\geq0}$ be the download/access request rate for the $i$-th data file, where $\mathbb{R}_{\geq0}$ denotes the set of positive real numbers. The $k$-tuple $(\lambda_1, \lambda_2, ..., \lambda_k)\in \mathbb{R}^k_{\geq 0}$ is considered as an achievable service rate if the system can serve it, while satisfying the capacity constraints on the nodes.
 
\begin{definition}\label{SRRdef}\cite{aktacs2021service}
    Consider a system where $k$ data files are distributed among $n$ servers using an $[n, k]$ linear code characterized by a generator matrix $G$ with the server capacity $\mu\geq 1$. Let $ \mathcal{R}=\{ \mathcal{R}_1,  \mathcal{R}_2, ...,  \mathcal{R}_k\}$ be a recovery $G$ system. Then the associated service rate region (SRR) is the set of all such $k$-tuples $(\lambda_1, \lambda_2, ..., \lambda_k)\in \mathbb{R}^k$ such that for each $i\in[k]$ there exists a list
    \[\{\lambda_{iR} \in \mathbb{R} \mid R \in \mathcal{R}_i, \lambda_{iR}\geq0\},\] satisfying the following conditions:

    \begin{equation} \label{eq1}
        \sum_{R\in \mathcal{R}_i}\lambda_{iR} =\lambda_i, \quad \forall i \in [k],
    \end{equation}
   \begin{equation}\label{eq2}
        \sum_{i=1}^k\sum_{\substack{ R\in \mathcal{R}_i\\ v \in R}}\lambda_{iR} \leq \mu, \quad \forall v \in [n].
    \end{equation} This service rate region is denoted by $\Lambda(\mathcal{R}(G), \mu)$.
\end{definition}
Equation (\ref{eq1}) guarantees that the demand for all data objects is adequately satisfied, while  (\ref{eq2}) asserts that the total number of requests handled by each server does not exceed its designated serving capacity. A collection of $\{\lambda_{iR}\}$ which satisfies \eqref{eq1} and \eqref{eq2} is called a valid allocation for the given recovery system $\mathcal{R}$ and service capacity $\mu$. Here $\lambda_{iR}$ is the fraction of $\lambda_i$ assigned to the recovery set $R$ of the $i$-th data symbol.

It is also important to note here that to determine the SRR for a given generator matrix $G$ it is not necessary to consider all the recovery sets, because $\Lambda(\mathcal{R}^{all}(G), \mu)$ is the same as $\Lambda(\mathcal{R}^{\min}(G), \mu)$\cite{alfarano2022dual}, where \[\mathcal{R}^{all}(G)=\{\mathcal{R}^{all}_1, \mathcal{R}^{all}_2, ..., \mathcal{R}^{all}_k\}.\] Additionally, $\Lambda(\mathcal{R}(G), \mu)=\mu\Lambda(\mathcal{R}(G), 1).$  

 For the rest of the paper, we will examine the service rate region for the minimum recovery $G$ system with uniform service capacity $\mu=1.$  We denote the corresponding SRR by $\Lambda(G).$
 
It is clear from Definition (\ref{SRRdef}) that the problem of finding the service rate region of a linear code is an optimization problem with $\lambda_{iR}$ as variables. However, as the value of $k$ increases, the complexity of solving this problem also increases due to the number of variables growing exponentially. To address this challenge, we turn to other mathematical tools to find the service rate region. 
Here, we briefly introduce two ways to interpret SRR due to \cite{aktacs2021service}: one is using the geometric representation and the other is utilizing the concept of hypergraphs.

\subsection{Geometric interpretation of SRR}
Determining the SRR of a coded scheme can be framed as a linear optimization problem where the intersection of several half-spaces encloses the SRR. The feasible region of this problem results in a convex polytope in $\mathbb{R}^k_{\geq 0}$. A comprehensive study of SRR, proving that it is a complex polytope, was done by Alfarano et al. \cite{alfarano2024service}, and we have included some of the basic concepts and results from their work here.
\begin{definition}
    A simplex in $\mathbb{R}^k$ is defined as the convex hull of the set $\{\delta_1 e_1, \delta_2 e_2, \ldots, \delta_k e_k\}$, where $\delta_i\in \mathbb{R}$ and $e_1, e_2, \cdots, e_k$ are the standard basis vectors in $\mathbb{R}^k$. In particular, a $\delta$-simplex is a simplex for which $\delta_i=\delta \text{ for all } i \in [k].$
\end{definition}
We define $\lambda_i^*(G)= \max\{x\mid xe_i \in \Lambda(G)\}$. Then, the largest possible size of the $\delta$-simplex in the SRR is given by  \[\delta(G)= \min \{\lambda_i^*(G), 1 \leq i\leq k\}.\]
The simplex structure plays a crucial role in characterizing SRR using the terms \textit{maximal achievable simplex} and \textit{maximal matching simplex}\cite{ly2025service}. The maximal matching simplex refers to the smallest geometric simplex that contains the SRR of the relevant storage scheme, while the maximal achievable simplex is the largest simplex contained in the SRR. Additionally, the term $\delta(G)$ represents the size of the maximal achievable $\delta$-simplex as it is given by the region of all service rate $k$-tuples satisfying $\sum_{i=1}^k\lambda_i\leq\delta$. 
\begin{remark}\label{avail}
    If there are $t$ disjoint recovery sets present in a coded system with at least one systematic server present in the generator matrix for every information symbol, then we have $ (t+1)e_i \in \Lambda(G), \forall i \in [k]$, implying that
    \[ \lfloor \delta(G) \rfloor \geq (t+1).\]
\end{remark}

\begin{lemma}\cite{kazemi2020geometric}\cite{alfarano2024service}\label{mindis}
    For an $[n, k, d]$ linear code with generator matrix $G$, if $\delta e_i\in\Lambda(G)$ for all $i\in[k]$, then $\lceil \delta \rceil \leq d.$ In particular, $\lceil \delta(G) \rceil \leq d.$
\end{lemma}

\subsection{Hypergraph as recovery \texorpdfstring{$G$}{G} system}

Hypergraphs are a generalization of graphs, where an edge (hyperedge) may contain more than two vertices. A hypergraph is represented by $H=(V, E(V))$, where $V$ is a set of vertices and $E(V)$ contains edges, and is a subset of the power set of $V$. The hypergraphs were first viewed as recovery $G$ systems in \cite{kazemi2020combinatorial}, where each hyperedge represents a recovery set and is labeled by the corresponding standard basis vector of the data file it recovers. The service rate served by an edge is associated with the fractional matching assigned to that edge. A matching in the hypergraph $H$ is a set of pairwise disjoint (non-adjacent) hyperedges. The matching number of a hypergraph $H$ is the maximum cardinality of a matching in $H$, denoted by $\nu(H)$. We can also consider matching as assigning a number $w(e)\in \{0, 1\}$ to each edge $e$ such that $\sum_{e \ni v} w(e)=1$ for each vertex $v$ of the hypergraph.

If we assign a number $w(e)$ to each edge $e$ such that $w(e) \in [0, 1]$ and $\sum_{e \ni v} w(e)\leq 1$ for each vertex $v$, then this setup is known as fractional matching. The maximum sum of the assigned 
number $w(e)$ of all the edges in the hyperedge $H$ is known as the \textit{fractional matching number}, and it is denoted by $\mu_f(H).$ Therefore, the problem of splitting the request rates to different recovery sets is now reduced to assigning a fractional matching to the hyperedges. The equivalence of a demand vector being in SRR and the existence of a corresponding fractional matching in the graph representation is proven in \cite{kazemi2020combinatorial}.

To optimally assign request rates, we need to maximize the fractional matching to the edges. 
For a recovery $G$ system and the associated recovery hypergraph $H$, the relation between the cumulative service rate and maximum fractional matching number is given by \[\sum_{i=1}^k\lambda_i\leq \mu_f(H).\]

Similar to the concept of vertex cover in graphs, we also have the vertex cover/transversal in hypergraphs. A transversal is the set $T\subseteq V$ such that for every $e\in E, T\cap e\neq \emptyset$. The smallest size of a transversal $T$ in $H$, denoted by $\tau(H)$, is known as the transversal number of $H$. For a hypergraph $H$, it stands that $\nu(H)\leq \mu_f(H)\leq\tau(H).$ Note that the size of the maximal matching simplex is equal to the fractional matching number $\mu_f(H)$.
\begin{lemma}\label{threeB}(\cite{kazemi2020combinatorial} Theorem 2)
    For a coding scheme with generator matrix $G$ and uniform service capacity $\mu=1$, if $(\lambda_1, \lambda_2, \ldots, \lambda_k)\in\Lambda(G)$, and $H$ is the associated recovery hypergraph, then the following relations are satisfied:
    \[\nu(H)\leq \max \left(\sum_{i=1}^k\lambda_i\right)=\mu_f(H)\leq\tau(H)~.\]
\end{lemma}
In the case where demand for some of the data symbols is zero, the service rate region is determined for the remaining data symbols. 
 The partial hypergraph $H_\mathcal{I}, \mathcal{I}\subseteq[k],$ is formed by removing those edges from $H$ which are not labeled by the standard basis vectors of the data symbols indexed by $\mathcal{I}$. This also satisfies the matching, fractional matching, and transversal number relation. We will use this fact to obtain new upper limits for the sum of service rates in our coding scheme.

The next subsection presents a demand-splitting technique also detailed in \cite{ly2025m}, and this technique is proven optimal for MDS codes in \cite{aktacs2017service}.

\subsection{A Greedy (Waterfilling) algorithm for request splitting}

 One of the critical challenges in finding the service rate regions of distributed data storage systems is designing a request-splitting strategy that optimally distributes user requests to storage nodes. The approach is inspired by the well-known greedy algorithm, which involves making the best choice at each step, with the expectation that this approach will lead to a globally optimal solution: Always serve a request by first attempting to assign it to the smallest available recovery set (or resource) that can satisfy it. This way we efficiently use the small size recovery group, keeping in mind the load at every node. 
 
 The waterfilling algorithm, as a special case of the greedy algorithm, and inspired by concepts in information theory and signal processing, has emerged as a promising approach to address this challenge. In the context of the service rate region of coded schemes, this algorithm was initially introduced in \cite{aktacs2017service} to characterize the SRR of systematic MDS codes.
 The algorithm allocates resources (such as power or bandwidth) among multiple channels to maximize overall system performance.

The waterfilling algorithm can be effectively adapted for distributed data storage by allocating user requests among storage nodes to minimize latency, balance load, and maximize throughput. It directs incoming requests to the systematic servers first, and once their capacity is reached, any remaining requests are sent to the least loaded recovery sets, excluding those with exhausted nodes. The load on a recovery set refers to the maximum load among the nodes it contains. 

This strategy comprehensively characterizes the SRR of systematic MDS codes for $n\geq 2k$. We utilize the same request-splitting approach for systematic binary Hamming codes, demonstrating that it provides the largest service rate region compared to alternative strategies.
  
  For non-systematic Hamming codes, we aim to use a greedy approach by directing requests to the smallest recovery sets with the least load. This involves keeping an index of all possible recovery sets for each object, sorting or organizing these sets by their cardinality, and potentially enhancing the approach with load-awareness to prevent hotspots.

 

\section{service rate region of Hamming Codes}\label{iii}

In this section, we study the service rate region of Hamming codes, particularly of the binary Hamming code $[2^r-1, k=2^r-1-r,3]$ denoted by Ham($r, 2$) where $r\geq2$. For $r=2$ the code Ham($2,2$) is the repetition code of length $3$ with $\Lambda=\{i\mid 0\leq i\leq 3)\}$. Therefore, the work in this paper focuses on the SRR of Ham($r,2$) for $r\geq3$.

We recall the definition of a $q$-ary Hamming code.
\begin{definition}\cite{huffman2003fundamentals}
    For $r\geq 2$, a linear code $C$ over a finite field $\mathbb{F}_q$ is called a $q$-ary Hamming code, denoted by Ham($r, q$), if the parity-check matrix of $C$ is defined by choosing for its columns a non-zero vector from each one-dimensional subspace of $\mathbb{F}_q^r.$ The code $C$ is a $\left[\frac{q^r-1}{q-1}, \frac{q^r-1}{q-1}-r, 3\right]$ linear code.
\end{definition}

The columns of a parity-check matrix of a binary Hamming code are all the non-zero vectors of $\mathbb{F}_2^r$ arranged in such a way that the $i$-th column is the binary expansion of the $i$-th number. Any other parity-check matrix of the binary Hamming code can be obtained by column permutation of the initial matrix. If the generator matrix of the Hamming code is systematic, then we say that the code is a systematic Hamming code.

Since there is exactly one systematic column (server) for each data symbol in a systematic generator matrix $G$ of Ham($r, q$), the $k$-tuple $(1, 1, \ldots, 1)$ is in the service rate region of systematic Ham($r, q$) for all $r$. The remaining $r$ parity columns of $G$ are of weight $q^{r-1}-1$. 

WLOG, we consider the standard form generator matrix for a systematic Ham($r, q$), represented as: 
\[G= \left(I_k, P\right)=(\underbrace{e_1, e_2, \ldots , e_k}_{\substack{\text{Standard basic} \\ \text{vectors in } \mathbb{F}_q^k}}, \underbrace{p_1 , p_2 , \ldots, p_r}_{\substack{\text{parity-check}\\ \text{columns}}}).\]
Throughout this section, we adopt this notation for the systematic generator matrix.

\begin{example}\label{ex1}
    The parity-check matrix for the $[7, 4, 3]$ $(r=3)$ binary Hamming code is as follows
\[H=  \left( \begin{array}{ccccccc}
  0&  0& 0 &1 &1 & 1& 1 \\
      0&  1& 1 &0 &0 & 1& 1\\
        1&  0& 1 &0 &1 & 0& 1\\
\end{array}\right) .\]
Then the corresponding generator matrix for Ham$(3, 2)$ is

\[G= \left( \begin{array}{ccccccc}
1 &1 & 1&   0&  0& 0 &0\\
      1 & 0&  0& 1 &1 &0 & 0\\
      0 &  1&  0& 1 &0 &1 & 0\\
      1 &1 & 0 & 1&   0&  0& 1\\
\end{array}\right).\]
The systematic generator matrix can be obtained by permuting the columns of this matrix.

\end{example}
The dual code of a Hamming code Ham($r, q$) is referred to as a simplex code S($r, q$) with the parameters $[\frac{q^r-1}{q-1}, r, q^{r-1}]$. Not only does the simplex code S($r, q$) have a minimum distance of $q^{r-1}$, but all the non-zero codewords in S($r, q$) have the same weight $q^{r-1}$. 

From Lemma \ref{Rset}, the support of codewords of a simplex code completely determines the structure and elements of the recovery sets of the data symbols involving the corresponding systematic Hamming coding scheme.

The following Theorem \ref{main} provides a general structure of these recovery sets in the systematic Ham($r, q$) coding scheme. 
\begin{theorem}\label{main}
    Let $C$ be a $\left[\frac{q^r-1}{q-1}, k= \frac{q^r-1}{q-1}-r,3\right]$ $q$-ary Hamming code with the systematic generator matrix $G$. Then the recovery $G$ system satisfies the following conditions:
    \begin{enumerate}
        \item A minimum recovery set has cardinality either $1$ or $q^{r-1}-1$.
        \item There are exactly $q^{r-1}$ non-singleton minimum local recovery sets of an information symbol.
        \item For an $i$-th information symbol and any other node $j$, the cardinality of the set $\{R\mid j\in R\}$, where $R\in\mathcal{R}^{min}_i$, is $(q-1)q^{r-2}$.
      
    \end{enumerate}
    \end{theorem}
\begin{proof}
    To prove the first statement, we use the fact that the minimum distance of the dual code of $C$ is $q^{r-1}$. Due to Lemma \ref{dperp}, the cardinality of a recovery set $R$ for the $i$-th information symbol is either $1$ or at least $q^{r-1}-1.$ As all the codewords of the dual code have weight $q^{r-1}-1$, we have $|R|=1$ or $|R|\leq q^{r-1}-1$. This proves the first statement.
    
    Any codeword of S($r, q$) with a non-zero entry at the $i$-th place is associated with a recovery set for the $i$-th coordinate. From Lemma \ref{ithdata}, we can say that there are exactly $q^{r-1}$ codewords in the simplex code S($r, q$) which have $1$ at the $i$-th place. These codewords provide exactly $q^{r-1}$ distinct recovery sets of cardinality $q^{r-1}-1$ for the $i$-th coordinate. Define 
    \[\mathrm{S}=\{\text{Supp}(c)\setminus\{i\}\mid c\in S(r, q), c_i=1\}.\]
    
    Since any codeword with $\alpha\in \mathbb{F}_q^*\setminus\{1\}$ at the $i$-th place is an $\mathbb{F}_q$-linear multiple of a codeword with $c_i=1$, considering the $i$-th server systematic, we have \[\mathrm{S}\cup\{i\}=\mathcal{R}_i^{\min}.\]
Therefore, the cardinality of the set containing non-singleton recovery sets of the $i$-th information node is $|\mathcal{R}_i^{\min}|-1=q^{r-1}.$

    Next, fix $i$, and consider all the codewords of S($r,q$) with the $i$-th coordinate zero; then the set of such codewords forms a subcode of S($r,q$), say $D$. The subcode $D$ has $q^{r-1}$ codewords, and \[\text{S}(r, q)=D \bigcup_{\substack{c_i=\alpha\\\alpha\in\mathbb{F}_q*}}(c+D).\]

The number of codewords in $D$ that have a zero in the $j$-th coordinate is $q^{r-2}$ by Lemma \ref{ithdata}. Consequently, the number of codewords in $D$ with a non-zero $j$-th coordinate is $q^{r-1} - q^{r-2}$. Furthermore, due to the linearity of the code $C$, each additive coset of $D$ also has $q^{r-1} - q^{r-2}$ codewords with a non-zero $j$-th coordinate. 
     
    Thus, the set of the support of the $q^{r-2}(q-1)$ codewords in the coset $c+D$ $(c_i=1)$, with the $j$-th coordinate non-zero, gives all the distinct recovery sets of the $i$-th data symbol containing the $j$-th node. Hence the result.
\end{proof}

\begin{remark}\label{nonG}
    For a non-systematic generator matrix $G$ of Ham($r, q$), recovery sets of data symbols can be of different sizes, but if there is a systematic server in $G$, then the recovery sets of the data symbol corresponding to this systematic column are determined in accordance with Theorem \ref{main}.
\end{remark}

\begin{example}\label{exr3}
    For $r=3$, we have the $[7,4,3]$ binary Hamming code whose generator and parity-check matrices are given in Example \ref{ex1}. Consider a message vector $u=(a, b, c, d)$. The associated codeword is $u\cdot G=(a+b+d, a+c+d, a, b+c+d, b, c, d)$. 
Then the recovery $G$ system for this code is given as follows:
\[\mathcal{R}^{min}(G)=\{\mathcal{R}_a, \mathcal{R}_b, \mathcal{R}_c, \mathcal{R}_d\}, \quad \text{where}\]
\[ 
\mathcal{R}_a=\{(3), (1,5,7), (2,6,7), (2,4,5), (1,4,6)\},\]
\[\mathcal{R}_b=\{(5), (1,3,7), (4,6,7), (2,3,4), (1,2,6)\},\]
\[\mathcal{R}_c=\{(6), (1,2,5), (2,3,7), (4,5,7), (1,4,3)\},\]\[\mathcal{R}_d=\{(7), (1,3,5), (2,3,6), (4,5,6), (1,2,4)\},\]
which satisfies the above theorem.
\end{example}
The next corollary is a consequence of Theorem \ref{main} together with Lemma \ref{mindis} and Remark \ref{avail}. The upper bound is due to the distance, and the lower bound is due to the availability.

\begin{corollary}
    In Ham($r,q$),  the size $\delta(G)$ of the maximal achievable $\delta$-simplex satisfies $\lceil\delta(G)\rceil \leq 3$. Moreover, in the case of systematic Hamming codes, we also have $\lfloor \delta(G) \rfloor \geq 2$.
\end{corollary}
 It is proven in \cite{ly2025m} that the upper bound $\delta(G)\leq 3$ is achievable for the systematic Ham($r, 2$) as the minimum weight codewords of S($r, 2$) form a $2$-design. We also demonstrate that the upper bound on $\delta(G)$ is indeed achievable for systematic Ham($r, 2$) through optimal request splitting. The method of directing incoming requests toward the recovery sets is crucial for this achievability.
\begin{lemma}In the systematic Ham($r, q$) with $r>3$, an optimal request-splitting strategy is first directing all incoming requests to the systematic servers upon arrival, and then assigning the remaining requests to the recovery sets with the current least load.

\end{lemma}
\begin{proof}
    Let $\lambda_i < \mu$ denote the incoming request rate for the $i$-th data object, and let $t_i$ represent the load assigned to the corresponding systematic server. Consequently, $\lambda_i-t_i$ is the load redirected to another recovery set (i.e., towards $q^{r-1} - 1$ nodes present in the recovery set). By redistributing the load in this manner, the systematic server's load is reduced by $\lambda_i-t_i$, while each of the $q^{r-1} - 1$ coded nodes absorbs an additional $\lambda_i-t_i$ load.  
    
Since $r > 3$, for Ham($r, q$), we have $n-k=r < q^{r-1}-1$, implying that at least one server in any non-singleton recovery set of the $i$-th data file is a systematic server for a different $j$-th data file. The capacity of the systematic server serving the $j$-th file is exhausted by $\lambda_i-t_i$; therefore, this redistribution of $\lambda_i-t_i$ reduces the maximum service rate for the $j$-th file on its systematic node by $\lambda_i-t_i$. Consuming the non-singleton recovery set's capacity potentially reduces maximum service for other data files. Thus, it is optimal to prioritize sending incoming requests to systematic servers whenever they have available capacity.

Once the systematic server for a data object is saturated, requests are directed towards the recovery set with the least load. Consider two non-singleton recovery sets $R_1$ and $R_2$ of a data symbol. Let the loads on the nodes in set $R_1$ be $\gamma_1 \geq\gamma_2\geq\cdots \geq\gamma_{q^{r-1}-1}$, and in the set $R_2$ be $\gamma_1' \geq\gamma_2'\geq\cdots\geq\gamma_{q^{r-1}-1}'$, such that $\gamma_1>\gamma_1'$. The set $R_2$, having the smaller load among the two, is capable of serving at most $\mu-\gamma_1'$ requests. This exceeds $\mu-\gamma_1$ requests that $R_1$ is capable of serving.
In particular, it cannot
be optimal to place the job on a more loaded recovery set
while a less loaded one exists. Hence, the greedy ``route-to-least-loaded” rule is optimal among recovery sets. If there are $m$ recovery sets with the same load, say $\gamma$, then a maximum of $\left(\frac{1-\gamma}{q^{r-2}}\right)m$ request rate is distributed uniformly across these $m$ recovery sets for maximum service.

Hence, the waterfilling algorithm ensures an optimal allocation of requests by first directing the requests to systematic servers and minimizing the impact on the overall service rate.
\end{proof}
The following result indicates that the size of the maximal achievable simplex for systematic Ham($r, q$) is $\left(1+\frac{q}{q-1}\right)$. It has very recently come to our knowledge that this result can be obtained from Theorem $1$ of \cite{ly2025m}, where the authors utilized a bound for the maximum achievable rate and demonstrated its achievability for the same coding scheme.
\begin{theorem}\label{maxd}
    In the SRR of the systematic $q$-ary Hamming code for all $r\geq 3$, the maximal achievable demand for a single data object is $\left(1+\frac{q}{q-1}\right)$, i.e., $\lambda_i^*=1+\frac{q}{q-1}$ for all $i\in[k]$.
\end{theorem}

\begin{proof}
    We show that the storage system with systematic Ham($r, q$) coding scheme is able to serve the request rates $\left(1+\frac{q}{q-1}\right)e_i$ for all $i\in[k]$. We use waterfilling algorithm to split the requests to the recovery sets. Since each server's capacity in our storage model is one, at most one request for an information symbol can be served using the corresponding systematic node.
    
From Theorem \ref{main}, we know that every node is in $q^{r-2}(q-1)$ recovery sets of a data object. The incoming request to download a data object is initially directed towards its systematic node. After the saturation of the systematic node, the optimal allocation is to send $\frac{1}{q^{r-2}(q-1)}$ of the request rate to every recovery set of the data object. This splitting strategy fulfills a maximum of $\frac{q^{r-1}}{q^{r-2}(q-1)}$ requests to download the required object without leaving any more service capacity in any node. 
Then we have 
\begin{equation*}
\left(1+\frac{q}{q-1},0,\ldots, 0\right), \left(0,1+\frac{q}{q-1},0,\ldots, 0\right), \ldots, \left(0,0, \ldots, 0,1+\frac{q}{q-1}\right)\in \Lambda(G)
\end{equation*}

Therefore, the maximum optimal allocation to serve a data object is $\lambda_i^*=1+\frac{q}{q-1}$ for any $i\in[k]$, when the request for the other data objects is zero.\end{proof}

 The result indicates that the maximum demand does not depend on the parameter $r$ but on the code alphabet. Since the service rate region is a convex polytope (see \cite{alfarano2024service}), the next corollary follows.

\begin{corollary} The convex hull of the set $\{(3,0,\ldots, 0), (0,3,0,\ldots, 0), \ldots,(0,\ldots,0,3)\} \subseteq \mathbb{R}^{2^r-1-r}$ lies in the SRR of the systematic binary Hamming code for all $r$, i.e., $\delta(G)=3$.
\end{corollary}
For a non-systematic matrix $G$ with a systematic column for some $i \in [k]$, a parallel argument can be used to conclude that $\lambda_i^* = \frac{2q-1}{q-1}$ when $\lambda_j = 0$ for all $j \neq i\in[k]$ as discussed in Remark \ref{nonG}. Additionally, the allocation of requests to access such a data symbol can also be done using the waterfilling algorithm to get the maximum output.

\subsection{Using Hypergraphs and matching to find SRR of systematic Ham(\texorpdfstring{$r,2$}{})}

In this section, we associate the binary Hamming codes with hypergraphs that illustrate the recovery structure of the code. In other words, the recovery $G$ system of Ham($r, 2$) is seen as a hypergraph with $(2^r-1)$ vertices. The recovery hypergraph is constructed with vertices corresponding to each column of the generator matrix $G$, where each recovery set of a data symbol forms a hyperedge.

Let the recovery hypergraph corresponding to the generator matrix $G$ be denoted as $\Gamma_G$. Each hyperedge (recovery set) of $\Gamma_G$ is labeled as the corresponding standard basis vector of the data symbol it recovers. A transversal and the transversal number of the hypergraph $\Gamma_G$ play a crucial role in proving the following theorem.

\begin{theorem}\label{odd}
    Suppose $G$ is a generator matrix for the code Ham($r,2$). Let $\Gamma_G$ be the associated recovery hypergraph and $O_w$ be the number of odd weight columns in $G$. For a $k$-tuple $(\lambda_1, \lambda_2, \ldots, \lambda_k)\in \Lambda(G)$, the following holds:
    \begin{enumerate}
        
       \item\label{1Ow} $\sum_{i=1}^k \lambda_i \leq O_w$,
        \item $\nu(\Gamma(G))\geq $ number of systematic columns in $G$,
        \item if $G$ is systematic, then  \begin{equation}\label{bound}
            \sum_{i=1}^k\lambda_i\leq \begin{cases}
            2^{r}-1-r, \quad \text{for } r>3, \\ 5, \quad \text{for } r=3,
        \end{cases}
        \end{equation} and this bound is achievable.
       
    \end{enumerate}
  
\end{theorem}
  \begin{proof}
        All the standard basis vectors of $\mathbb{R}^k$ are of weight one. In order to obtain these vectors as a sum of other column vectors of $G$, there must be at least one column vector with an odd weight in each recovery set. This leads us to conclude that the intersection of any hyperedge of $\Gamma(G)$ with the set of odd-weighted column vectors of $G$ is non-empty. Since the transversal number is the smallest number of vertices needed to cover all hyperedges, we have $\tau(\Gamma(G))\leq O_w.$ Then, by Lemma \ref{threeB}, the inequality in \eqref{1Ow} follows. 
        
        However, in the case of a systematic generator matrix, we have $O_w=n,$ which does not lead to any concrete results.
        
The set of all systematic nodes in $G$ forms a set of disjoint edges of the recovery hypergraph $\Gamma(G)$. Therefore, the matching number is greater than or equal to the number of systematic nodes in $G$.

        If $G$ is systematic, then we have each systematic node as an edge of the hypergraph, leading to the containment of all the systematic nodes in any vertex cover. According to Theorem \ref{main}, all the non-singleton recovery sets have cardinality $2^{r-1}-1$, which is greater than $r$ for $r>3$. Therefore, there is no recovery set which contains only non-systematic columns. Consequently, it is impossible to form a recovery set that consists solely of non-systematic columns. This means that the set of all $k$ systematic columns constitutes a vertex cover for the associated hypergraph, resulting in $\tau(\Gamma(G))=k=2^r-1-r,$ for $r>3.$ To achieve the bound in ($\ref{bound}$), for each data symbol, a request is sent to the corresponding systematic node and served by it.
        
        However, in the case of $r=3$, the structure of the parity-check matrix allows for a recovery set made entirely of the three non-systematic nodes. As a result, any vertex cover includes one non-systematic column along with all systematic columns, leading to $ \tau(\Gamma(G)) = k + 1 = 5 $. Thus, the bound in ($\ref{bound}$) is achieved by serving one request for each data symbol by the corresponding systematic node, and serving one additional request through the recovery set with all the non-systematic nodes.
 This completes the proof.
  \end{proof} 

  \begin{corollary}
      For a non-systematic generator matrix $G$ of Ham$(3,2)$, the size of the maximal matching simplex for $\Lambda(G)$ is $O_w$.
  \end{corollary}
 \begin{remark}\cite{aktacs2021service}(Proposition 5)
    The result \ref{1Ow} of Theorem \ref{odd} can also be concluded considering the hyperplane $\mathcal{H}$ of the projective space $PG_n(\mathbb{F}_2)$, defined by the equation
\[
x_1 + \cdots + x_k = 0.
\]

It is clear that none of the standard basis vectors $e_i$ is in $\mathcal{H}$, as their weight is $1$ (odd). Moreover, $|\mathcal{H}|=n-O_w$, and for any $\lambda \in \Lambda(G)$, it yields 
\[
\sum_{i=1}^{k} \lambda_i \leq |G\setminus\mathcal{H}|=O_w.
\]

 \end{remark}

\begin{example}\label{NSExample}
     Let $(a, b, c, d)$ be the data symbols and let $G$ be a generator matrix of the Ham$(3, 2)$ code, given as follows:\[G= \left( \begin{array}{ccccccc}
1 &1 & 0&   0&  1& 1 &0\\
      0 & 0&  1& 0 &1 &1 & 0\\
      1& 0 &  1&  0& 1 &0 &1 \\
      0 &1 & 1 & 1&   1&  0& 0\\
\end{array}\right).\] 
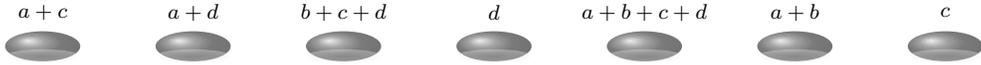
\begin{figure}[H]
    \centering
    
    \begin{tikzpicture}
    \def\labels{{"$a+c$","$a+d$","$b+c+d$","$d$","$a+b+c+d$","$a+b$","$c$"}}
    \foreach \i [count=\n from 0] in {0,1,2,3,4,5,6} {
       
        \shade[ball color=gray!40!black, opacity=0.6] (\i*2,0) ellipse (0.5 and 0.20);
       
        \shade[inner color=black!30,outer color=white,opacity=0.3] (\i*2, -0.13) ellipse (0.5 and 0.09);
        \node[font=\footnotesize] at (\i*2,0.45) {\pgfmathparse{\labels[\n]}\pgfmathresult};
    }
\end{tikzpicture}

    \caption{Distributed storage system storing $4$ files on $7$ servers (nodes) using binary Hamming scheme (non-systematic generator matrix).}
    \label{ellipse}
\end{figure}
The odd weight columns of $G$ are columns $3, 4$, and $7$. The recovery $G$ system or the sets of edges of the recovery hypergraph $\Gamma(G)$ are presented below:
\begin{align*}
    \mathcal{R}_a&=\{(1, 7), (2, 4), (3, 5), (1, 4, 5, 6), (2, 5, 6, 7), (3, 4, 6, 7), (1, 2, 3, 6)\},\\
\mathcal{R}_b &=\{(3, 4, 7), (1, 6, 7), (3, 5, 6), (2, 4, 6), (2, 5, 7), (1, 4, 5), (1, 2, 3)\},\\
\mathcal{R}_c&=\{( 7), (1, 2, 4), (1, 3, 5), (4, 5, 6), (2, 3, 6)\},\\
\mathcal{R}_d &=\{(4), (1, 2, 7), (2, 3, 5), (5, 6, 7), (1, 3, 6)\}.
\end{align*}

Since the set $\{3, 4, 7\}$ is a vertex cover of $ \Gamma(G)$, the transversal number and the matching number are $\tau(\Gamma(G))=3$, and $\nu(\Gamma(G))=3$, respectively. It follows from Lemma \ref{threeB} that $\sum_{i=1}^4\lambda_i\leq3$, and this bound is achieved by serving requests from edges $(3, 5),(7),$ and $(4)$. 

The requests are split using the greedy technique and are first directed towards the recovery sets of smallest cardinality with the least load for each of the data symbols. For data file $a$, the maximum demand served is $\lambda_a^{\max}=3$. Similarly, for files $c$ and $d$, we also have $\lambda_c^{\max}=\lambda_d^{\max}=3.$

Since in $\mathcal{R}_b$ each node is contained in three edges, to get maximum service for file $b$, $1/3$ request is sent to each edge, recovering $\frac{7}{3}$ of file $b$.
Note that $\delta(G)\ne3,$ since $\lambda_b^{\max}$ achievable is $\frac{7}{3}$.

\end{example}

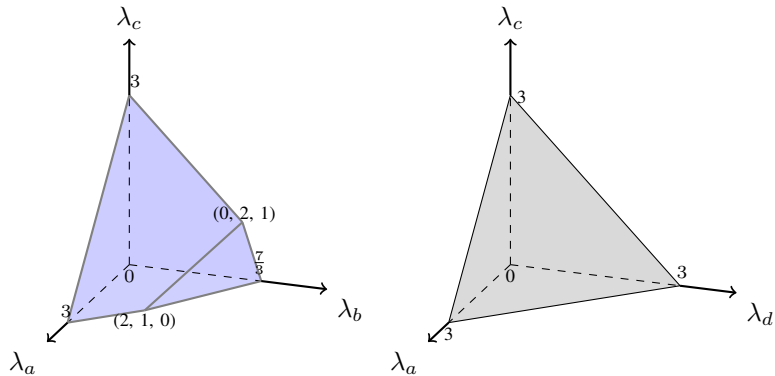
\begin{figure}[ht]
  \tdplotsetmaincoords{70}{110}
  \centering
 \begin{tikzpicture}[tdplot_main_coords, scale=0.8]
\begin{scope}[shift={(0,0,0)}]

\fill[blue!20, opacity=0.7] (3,0,0) -- (2,1,0) -- (0,{7/3},0) -- (0,0,0) -- cycle;
\fill[blue!30, opacity=0.7] (0,0,0) -- (0,{7/3},0) -- (0,2,1) -- (0,0,3) -- cycle;
\fill[blue!25, opacity=0.7] (0,0,0) -- (0,0,3) -- (3,0,0) -- cycle;
\fill[blue!20, opacity=0.7] (3,0,0) -- (0,0,3) -- (0,2,1) -- (2,1,0) -- cycle;
\fill[blue!20, opacity=0.7] (2,1,0) -- (0,2,1) -- (0,{7/3},0) -- cycle;

    \coordinate (Xfix) at (3,0,0);
    \coordinate (Yfix) at (0,7/3,0);
    \coordinate (Zfix) at (0,0,3);

    \draw[dashed] (0,0,0) -- (Xfix);
   \draw[dashed] (0,0,0) -- (Yfix);
   \draw[dashed] (0,0,0) -- (Zfix);

\draw[thick, gray] (3,0,0) -- (0,0,3);
\draw[thick, gray] (0,{7/3},0) -- (2,1,0);
\draw[thick, gray] (0,{7/3},0) -- (0,2,1);
\draw[thick, gray] (0,0,3) -- (0,2,1);
\draw[thick, gray] (2,1,0) -- (3,0,0);
\draw[thick, gray] (2,1,0) -- (0,2,1);

\node[above right] at (3.5,-0.10,0.1) {\scriptsize$3$};
\node[above left] at (0,2.6,0) {\scriptsize$\frac{7}{3}$};
\node[above] at (0,0.1,3) {\scriptsize$3$};

\node at (0,0,-0.2) {\scriptsize$0$};
\node at (2,1,-0.2) {\scriptsize (2, 1, 0)};
\node at (-0.1,2,1.1) {\scriptsize (0, 2, 1)};

    \draw[thick,->] (Xfix) -- (4,0,0) node[anchor=north east] {$\lambda_a$};
    \draw[thick,->] (Yfix) -- (0,3.5,0) node[anchor=north west] {$\lambda_b$};
    \draw[thick,->] (Zfix) -- (0,0,4) node[anchor=south]{$\lambda_c$};
\end{scope}
\end{tikzpicture}
  \begin{tikzpicture}

\tdplotsetmaincoords{70}{110}
\begin{scope}[tdplot_main_coords, scale=0.8]
    \coordinate (A) at (0,0,0);
    \coordinate (B) at (3,0,0);
    \coordinate (C) at (0,3,0);
    \coordinate (D) at (0,0,3);
    
    \fill[gray!10, opacity=0.9] (A)--(B)--(C)--cycle;
    \fill[gray!20, opacity=0.9] (A)--(B)--(D)--cycle;
    \fill[gray!30, opacity=0.9] (B)--(C)--(D)--cycle;
  
 \fill[gray!30, opacity=0.9] (A)--(C)--(D)--cycle;

   
    \draw (D)--(B)--(C)--cycle;

    \node at (0,0,-0.2) {\scriptsize 0};
    \node at (0,0.2,3) {\scriptsize 3};
    \node at (-0.1,3,0.2) {\scriptsize ${3}$};
    \node at (3,0,-0.2) {\scriptsize 3};
     
    \coordinate (Xfix) at (3,0,0);
    \coordinate (Yfix) at (0,3,0);
    \coordinate (Zfix) at (0,0,3);

    \draw[dashed] (A) -- (Xfix);
   \draw[dashed] (A) -- (Yfix);
   \draw[dashed] (A) -- (Zfix);

    \draw[thick,->] (Xfix) -- (4,0,0) node[anchor=north east] {$\lambda_a$};
    \draw[thick,->] (Yfix) -- (0, 4,0) node[anchor=north west] {$\lambda_d$};
    \draw[thick,->] (Zfix) -- (0,0,4) node[anchor=south]{$\lambda_c$};
\end{scope}
\end{tikzpicture}
\hspace{0.5cm}

    \caption{The $3$D-projection of service polytope considering $\lambda_d=0$ (top), and service polytope when $\lambda_b=0$ (bottom) in Example \ref{NSExample}.}
    \label{ExSrr}
\end{figure}

Now we direct our attention to the service rate region of the systematic Ham($r,2$), and generalize the second point of Theorem \ref{odd} for any subset of $[k]$, under the notion that the cumulative achievable service rate is less than the transversal number of the associated recovery hypergraph. The recovery hypergraph $\Gamma_G$ is defined by considering all data symbols and their recovery sets. Now, we consider partial hypergraph $\Gamma_{G_\mathcal{I}}$ of $\Gamma_G$ with respect to a subset $\mathcal{I}\subseteq [k]$ such that the hyperedges labeled with $e_j$, where $j\in[k]\setminus \mathcal{I}$ are removed from $\Gamma(G)$. This formulation is analogous to considering that the demand/service rates for the data symbols indexed by nodes in $[k] \setminus \mathcal{I}$ are zero. 

The fractional matching number of an $s$-uniform hypergraph with $n$ vertices is less than or equal to $\frac{n}{s}$. We use this result as described in \cite{aktacs2021service}, and add $(2^{r-1} - 2)$ dummy vertices to each systematic node to represent the associated hyperedge in the recovery hypergraph (or partial hypergraph). As a result, the recovery hypergraph becomes $(2^{r-1}-1)$-uniform. 

The following theorem is a consequence of the above discussion.

\begin{theorem}
    Suppose $G$ is the systematic generator matrix of Ham($r, 2$), and $(\lambda_1, \lambda_2, \ldots, \lambda_k)\in \Lambda(G)$. Let $\mathcal{I}\subseteq [k=2^r-1-r]$. Then \[\sum_{i\in\mathcal{I}}\lambda_i\leq  |\mathcal{I}|+2-\frac{\mathcal{|I|}-1}{2^{r-1}-1}.\]
\end{theorem}

Since the fractional part remains less than $1$ for $|\mathcal{I}|>1$, the size $\delta_{m_\mathcal{I}}$ of maximal matching simplex corresponding to the data symbols indexed by elements of $\mathcal{I}$ satisfies $\lceil\delta_{m_\mathcal{I}}\rceil\leq |\mathcal{I}|+1$. 
We have shown later that this bound is attainable in numerous instances, and notably, it can be achieved without the nearest integer approximation.

 The proposed lemma yields the number of recovery sets that contain a fixed number of non-systematic nodes of the systematic Ham($r,2$). The recovery sets with cardinality $(2^{r-1}-1)$ include both systematic and non-systematic nodes for $r>3$.
 \begin{lemma}
          Suppose $G$ is the systematic generator matrix of a binary Hamming code Ham($r, 2$), $r\geq3$. If a recovery set contains $t$ $(\leq r)$  non-systematic nodes, then the total number of such recovery sets is $\binom{r}{t}(2^{r-1}-t)$.

 \end{lemma}
\begin{proof}
 Let $R_a$ be a non-singleton recovery set with $n_1, n_2, \cdots, n_t$ non-systematic nodes for some data symbol $a$. Since it contains $t$ non-systematic nodes, the remaining $2^{r-1}-1-t$ nodes are systematic. In other words, the support of the sum of these $t$ non-systematic columns has cardinality $2^{r-1}-t$. Let \begin{align*}
     S_p& =\text{Supp}\left(\sum (c_{n_1}+c_{n_2}+\cdots+c_{n_t})\right)\\ &=\{s_1, s_2, \ldots, s_{2^{r-1}-t}\}\subset [k].
 \end{align*} 
    
    Among the nodes in $S_p$, any $(2^{r-1}-1-t)$ are the chosen systematic nodes to form a recovery set for the remaining data symbol. Thus, a recovery set of $s\in S_p\setminus \{s_{i_1}$, $s_{i_2}, \ldots, s_{i_{(2^{r-1}-t-1)}}\}$ has the form
    \begin{equation}\label{form}
        \{n_1, n_2, \ldots, n_t, s_{i_1}, s_{i_2}, \ldots, s_{i_{(2^{r-1}-t-1)}}\}.
    \end{equation} Therefore,  
    the total number of recovery sets with $t$ non-systematic nodes is $\binom{r}{t}(2^{r-1}-t)$.\end{proof}

 If the set $\mathcal{I}$ is chosen in such a way that in the associated partial hypergraph there is a hyperedge that does not contain any nodes from $\mathcal{I}$ or has only non-systematic nodes, then any vertex cover of the partial hypergraph cannot consist only of the nodes present in the set  $\mathcal{I}$. For instance, consider the $2^{r-1}-r$ recovery sets of the form (\ref{form}) with $r$ non-systematic nodes and $2^{r-1}-r-1$ systematic nodes. The sum of the columns of $G$ corresponding to these $r$ non-systematic nodes is a vector of weight $2^{r-1}-r$ in $\mathbb{F}_2^k$. Denote this vector by $c_r$. If for some $s\in$ Supp($c_r$)$\cap\mathcal{I}$, the set Supp($c_r$) $\setminus \{s\}$ is not in $\mathcal{I}$, then there exists a recovery set say $R_s$ of $s$ which has $r$ non-systematic nodes and $(2^{r-1}-r-1)$ systematic nodes present in Supp($c_r$) $\setminus \{s\}$. Therefore, a vertex cover of the partial hypergraph $\Gamma_{G_\mathcal{I}}$ also has an element from $[n]\setminus\mathcal{I}$ to cover the hyperedge associated with the recovery set $R_s$, along with the set $\mathcal{I}$.

 The above argument raises a question about how large a vertex cover of the partial hypergraph  $\Gamma_{G_\mathcal{I}}$ is, compared to the set $\mathcal{I}$. The following results aim to address this question.

\begin{theorem}\label{2node}
    Let $G$ be a generator matrix of Ham($r, 2$) such that $G$ contains systematic columns corresponding to two information nodes $i$ and $j$. Then for a service request $\lambda\in \Lambda(G)$, \[\lambda_i+\lambda_j\leq3.\]
    \end{theorem}
\begin{proof}
    Let $\mathcal{I}=\{i, j\}.$ In the partial recovery hypergraph $\Gamma_{G_\mathcal{I}}$, there are two edges $\{i\}$ and $\{j\}$, each of which is assigned a weight $1$. Then the capacity of these two nodes is exhausted. Therefore, now the hyperedges labeled $e_i$ containing node $j$ cannot be used for more request allotment. Similarly, the hyperedges labeled $e_j$ containing node $i$ cannot be used.

   Let $v_m$ represent the $m$-th column of the parity-check matrix of $C$ corresponding to the matrix $G$. Note that each hyperedge labeled $e_i$ containing the vertex $j$ corresponds to a codeword $c = (c_1, c_2, \ldots, c_n)\in C^\perp$ such that $c_i =c_j = 1 $. 
   Since all the columns of the parity-check matrix $H$ of Ham($r, 2$) are distinct and linearly independent, for $i$ and $j$, there is always a node $s \in [n]$ such that $ v_i + v_j = v_s$. Consequently, in the partial hypergraph, there are $ 2^{r-2}$ hyperedges labeled $e_i$ that contains $s$ but not $j$, and there are also $2^{r-2}$ hyperedges labeled $e_i$ that contains $j$ but not node $s$. The same can be said for the hyperedges labeled $e_j$ in the context of nodes $i$ and $s$.
   
    After assigning weight $1$ to each node $i$ and $j$, there are $2^{r-2}$ hyperedges $e_i$ each containing node $s$, and $2^{r-2}$ hyperedges $e_j$ each containing node $s$, along with all nodes having unused capacity. Since each available recovery set contains $s$, at most one request can be served in this system. A vertex cover for $\Gamma_{G_\mathcal{I}}$ is given by the set $\{i, j, s\}$. As a result, we conclude that
 $\lambda_i+\lambda_j\leq3$.
  \end{proof}

 \begin{remark}
     If the node $s$ which is uniquely determined by the nodes $i$ and $j$, is also a systematic node in $G$, then $\lambda_i+\lambda_j+\lambda_s\leq3.$
 \end{remark}
The argument of Theorem \ref{2node} has been extended to provide bounds on the sum of service rates for systematic Ham($r, 2$) corresponding to the data symbols indexed by any subset of $[k]$. In this theorem, we discuss the characteristics of the SRR for systematic binary Hamming codes and prove the achievability of the identified bounds. 
   \begin{theorem}\label{main2}
       Suppose $G$ is the systematic generator matrix of a binary Hamming code Ham($r, 2$) where $r>3,$ and $H$ is the corresponding parity-check matrix. Let $\mathcal{I}$ be a set consisting of information nodes (not all) with $\mathcal{I}\geq2$, and $v_j$ be the $j$-th column of $H$. If the demand for all information symbols not indexed by elements in $\mathcal{I}$ is zero, then the size $\delta_{m_\mathcal{I}}$ of the maximal matching simplex is given by:
\[\sum_{i\in \mathcal{I}}\lambda_i\leq \delta_{m_\mathcal{I}}=\begin{cases}
    
    |\mathcal{I}|, \quad \text{if } \sum_{j\in \mathcal{I}}v_j=0 \in \mathbb{F}_2^r,\\ |\mathcal{I}|+1, \quad \text{otherwise}. 
\end{cases}\]
Additionally, if the service rate for all the information symbols is non-zero, i.e., $\mathcal{I}=[k]$, then $\delta_{m_\mathcal{I}}=k.$
   \end{theorem}
   \begin{proof}
       The columns of $H$ are all distinct. Also, since every non-zero vector of $\mathbb{F}_2^r$ is a column of $H$, for any subset $\mathcal{I}$ of information nodes, the sum $\sum_{j\in \mathcal{I}} v_j$ is again a column of $H$ or the all-zero vector.
       
Consider  $\sum_{j\in \mathcal{I}}v_j=0.$ This indicates that in any recovery set of an $i$-th data symbol, $i\in\mathcal{I}$, there is always at least one another node $s\in\mathcal{I}$ present in that recovery set. Moreover, all the hyperedges of the partial hypergraph $\Gamma_{G_\mathcal{I}}$ always contain at least one element of $\mathcal{I}$. This implies that the set $\mathcal{I}$ forms a vertex cover of $\Gamma_{G_\mathcal{I}}$. There are $|\mathcal{I}|$ distinct hyperedges in $\Gamma_{G_\mathcal{I}}$ which has the form $\{i\}$ with $i\in\mathcal{I}$, hence this set is the smallest vertex cover of $\Gamma_{G_\mathcal{I}}$ and $\tau(\Gamma_{G_\mathcal{I}})=|\mathcal{I}|=\nu(\Gamma_{G_\mathcal{I}})$. 

 Now, suppose $\sum_{j\in \mathcal{I}}v_j\neq 0$. Then there exists a column of $H$, say $v_s$, such that \begin{equation}\label{nodeS}
     \sum_{j\in\mathcal{I}}v_j=v_s.\end{equation} If a codeword $c=(c_1, c_2, \cdots, c_n)\in C^\perp$ with $i\in$ Supp($c$) and $s\notin$ Supp($c$) for some $i\in\mathcal{I}$, then $|\text{Supp}(c) \cap \mathcal{I}|\geq2$. Thus, the hyperedges associated with such codewords of the dual code are covered by $\mathcal{I}$. 
 
  Now we consider those hyperedges of $\Gamma_{G_\mathcal{I}}$ which do not contain any vertex from the set $\mathcal{I}$. These hyperedges correspond to the codewords $c\in C^\perp$ such that $|\text{Supp}(c) \cap \mathcal{I}|=1$.  Let $h_{c_a}\in \mathcal{R}_a$ be the hyperedge of $\Gamma_{G_\mathcal{I}}$ corresponding to a codeword $c\in C^\perp$ for $ \text{Supp}(c) \cap \mathcal{I}=\{a\}$. 

 \begin{itemize}
     \item{\textbf{Case 1:}} If the codeword $c$ is a row of the parity-check matrix, then considering Equation (\ref{nodeS}), we have $c_s=1$. Thus, $s\in h_{c_a}$ and $\mathcal{I}\cup\{s\}$ is the smallest transversal of $\Gamma_{G_\mathcal{I}}$.

     \item{\textbf{Case 2:}}
If the codeword $c\in C^\perp$ is a linear combination of rows $\{x_1, x_2, \ldots, x_t\}, 1\leq t\leq r,$ of the parity-check matrix $H$. Then, from the assumption $\text{Supp}(c) \cap \mathcal{I}=\{a\}$, we have \begin{equation}\label{a}
    x_{1a}+x_{2a}+\cdots+x_{ta}=1,
\end{equation}
\begin{equation}\label{b}
    \sum_{i=1}^tx_{ij}=0, \text{ for all } j\in \mathcal{I}\setminus \{a\},
\end{equation}
where $x_{ij}$ is the $j$-th coordinate of the row vector $x_i$.
After adding equations \eqref{a} and \eqref{b} we get \begin{align}
     \sum_{j\in \mathcal{I}} \sum_{i=1}^tx_{ij}=\sum_{i=1}^t\sum_{j\in \mathcal{I}}x_{ij}&=1\label{c}.\end{align} 
Now, applying the relation given in Equation \eqref{nodeS} to Equation \eqref{c}, we get \begin{align*}
   \sum_{i=1}^t\sum_{j\in \mathcal{I}}x_{ij}= \sum_{i=1}^t x_{is}=c_s&=1.
\end{align*}

This indicates that the node $s$ is in the hyperedge $h_{c_a}$. In this case as well, we conclude that the smallest vertex cover of the partial hypergraph $\Gamma_{G_{\mathcal{I}}}$ is formed by the set $\mathcal{I}\cup\{s\}$.
 \end{itemize}

 Note that the matching number of $\Gamma_{G_{\mathcal{I}}}$ in both the cases is $|\mathcal{I}|+1$.
 Therefore, the size $\delta_{m_\mathcal{I}}$ of the maximal matching simplex is $|\mathcal{I}|+1$.

For the systematic Ham($r, 2$) with $r > 3$ and the set $\mathcal{I} = [k]$, it is evident that all singleton hyperedges are included in a vertex cover. Also, there are no hyperedges that do not contain a systematic server; the set $[k]$ forms a complete smallest vertex cover of $\Gamma(G)$. Therefore, we conclude that $\delta_{m_\mathcal{I}} = k$.
\end{proof}
Since any two columns of the parity-check matrix of Ham($r, 2$) are linearly independent, Theorem \ref{2node} is a special case of Theorem \ref{main2}. The fractional matching number is equal to the transversal number for $\Gamma(G)$ of the systematic Ham($r, 2$), therefore, if the cumulative service rate exceeds the integral value $\tau(\Gamma(G))$, then the system is not able to complete the requested demand.

\begin{corollary}
    Let $r>3$ be a positive integer, and consider a $k$-tuple $(\lambda_1, \lambda_2, \ldots, \lambda_k)$, where $\lambda_i=1$ for all $i\in[k]$ except one index $j\in[k]$, for which $\lambda_j>1$. Then, the $k$-tuple $(\lambda_1, \lambda_2, \ldots, \lambda_k)$ does not belong to the service rate region of systematic Ham($r, 2$).
\end{corollary}
The next theorem summarizes the discussion presented above.

\begin{theorem}
    Let $G$ denote the systematic generator matrix and $H$ be the corresponding parity-check matrix for Ham($r, 2$) with the parameters $[2^r - 1, k = 2^r - 1 - r, d = 3]$. Consider a subset $\mathcal{I}$ of the set of systematic nodes in $G$ such that $2<|\mathcal{I}| \neq k$, and the demand for all data symbols indexed by elements in $[k] \setminus \mathcal{I}$ is zero. Let $v_j$ denote the $j$-th column of $H$. Then the service rate region for the systematic Ham$(r, 2)$ is given by:
\[\begin{cases}
    \lambda_i\geq 0, \quad \text{ for all }i\in[k],\\
    \lambda_i+\lambda_j \leq 3 \quad \text{ for all }i, j\in[k],\\
    \sum_{i\in \mathcal{I}}\lambda_i\leq \delta_{m_\mathcal{I}}=\begin{cases}
    |\mathcal{I}|, \quad \text{if } \sum_{j\in \mathcal{I}}v_j=0 \in \mathbb{F}_2^r,\\|\mathcal{I}|+1, \quad \text{otherwise},
    \end{cases}\\
    \sum_{i=1}^k\lambda_i\leq \begin{cases}
            k, \quad \text{for }r>3, \\ 5, \quad \text{for }r=3.
        \end{cases}
     \end{cases}\]
    \end{theorem}

\begin{example}\label{finalex}
    Let $G$ be the matrix as described in Example \ref{exr3}. Then the service rate region of systematic Ham($3, 2$) is the feasible region of the linear programming problem with the following constraints:
\begin{align*}
   \lambda_a+\lambda_b+\lambda_c+\lambda_d&\leq5, \\
   \lambda_a+\lambda_b+\lambda_c&\leq3,\\
   \lambda_i+\lambda_j+\lambda_d &\leq4, \quad i\neq j\in\{a, b, c\}, \\
    \lambda_i+\lambda_j&\leq3, \quad i\neq j\in\{a, b, c, d\},\\
    0\leq \lambda_i&\leq3, \quad \forall  i\in \{a, b, c, d\}.
\end{align*}
    Two $3$-dimensional cross-sections of the corresponding $4$-dimensional service polytope are illustrated in Figure \ref{SRR}.
\end{example}

\begin{figure}[ht]
  \centering
\begin{tikzpicture}

\tdplotsetmaincoords{70}{110}
\begin{scope}[tdplot_main_coords, scale =0.8, shift={(0,0,0)}]
    \coordinate (A) at (0,0,0);
    \coordinate (B) at (3,0,0);
    \coordinate (C) at (0,3,0);
    \coordinate (D) at (0,0,3);
    \fill[blue!20, opacity=0.9] (A)--(B)--(C)--cycle;
    \fill[blue!20, opacity=0.9] (A)--(B)--(D)--cycle;
    \fill[blue!20, opacity=0.9] (A)--(C)--(D)--cycle;
    \fill[blue!20, opacity=0.9] (B)--(C)--(D)--cycle;
    \draw (B)--(C);
    \draw (D)--(B);
    \draw (C)--(D);
   
    \node at (0,0,-0.2) {\scriptsize 0};
    \node at (0,0.2,3) {\scriptsize 3};
    \node at (-0.1,3,0.2) {\scriptsize 3};
    \node at (3,0,-0.2) {\scriptsize 3};
     
    \coordinate (Xfix) at (3,0,0);
    \coordinate (Yfix) at (0,3,0);
    \coordinate (Zfix) at (0,0,3);

    \draw[dashed] (A) -- (Xfix);
   \draw[dashed] (A) -- (Yfix);
   \draw[dashed] (A) -- (Zfix);
     
    \draw[thick,->] (Xfix) -- (4.3,0,0) node[anchor=north east] {$\lambda_a$};
    \draw[thick,->] (Yfix) -- (0,4.3,0) node[anchor=north west] {$\lambda_b$};
    \draw[thick,->] (Zfix) -- (0,0,4.3) node[anchor=south]{$\lambda_c$};
\end{scope}
\end{tikzpicture}
\hspace{0.5cm}
\begin{tikzpicture}

\tdplotsetmaincoords{70}{110}
\begin{scope}[tdplot_main_coords, scale=0.8, shift={(11,0,0)}]
    \coordinate (A) at (0,0,0);
    \coordinate (B) at (3,0,0);
    \coordinate (C) at (0,3,0);
    \coordinate (D) at (0,0,3);
    \coordinate (E) at (1,1,2);
    \coordinate (F) at (2,1,1);
    \coordinate (G) at (1,2,1);
    \fill[gray!30, opacity=0.9] (A)--(B)--(C)--cycle;
     \fill[gray!30, opacity=0.9] (A)--(B)--(D)--cycle;
    \fill[gray!30, opacity=0.9] (A)--(C)--(D)--cycle;
    \fill[gray!30, opacity=0.9] (B)--(D)--(E)--(F)--cycle;
    \fill[gray!30, opacity=0.9] (B)--(C)--(F)--(G)--cycle;
    \fill[gray!30, opacity=0.9] (C)--(D)--(E)--(G)--cycle;
    \fill[gray!20, opacity=0.9] (E)--(F)--(G)--cycle;

    \coordinate (Xfix) at (3,0,0);
    \coordinate (Yfix) at (0,3,0);
    \coordinate (Zfix) at (0,0,3);

    \draw[dashed] (A) -- (Xfix);
   \draw[dashed] (A) -- (Yfix);
   \draw[dashed] (A) -- (Zfix);
    \draw(B)--(C);
    \draw (D)--(B);
    \draw (C)--(D);
\draw (E)--(G);
   
    \draw (D)--(E)--(F)--(B);
    \draw (F)--(G)--(C);

    \node at (0,-0.2,0) {\scriptsize 0};
    \node at (0,0.2,3) {\scriptsize 3};
    \node at (-0.1,3,0.2) {\scriptsize 3};
    \node at (3,0,-0.2) {\scriptsize 3};
    \draw[thick,->] (Xfix) -- (4.3,0,0) node[anchor=north east] {$\lambda_a$};
    \draw[thick,->] (Yfix) -- (0,4.3,0) node[anchor=north west] {$\lambda_b$};
    \draw[thick,->] (Zfix) -- (0,0,4.3) node[anchor=south]{$\lambda_d$};
\end{scope}

\end{tikzpicture}

    \caption{The service polytope considering (i) the index set $\mathcal{I}=\{a, b, c\}$ (top), and (ii) when the index set is $\mathcal{I}=\{a, b, d\}$ (bottom).}
    \label{SRR}
\end{figure}
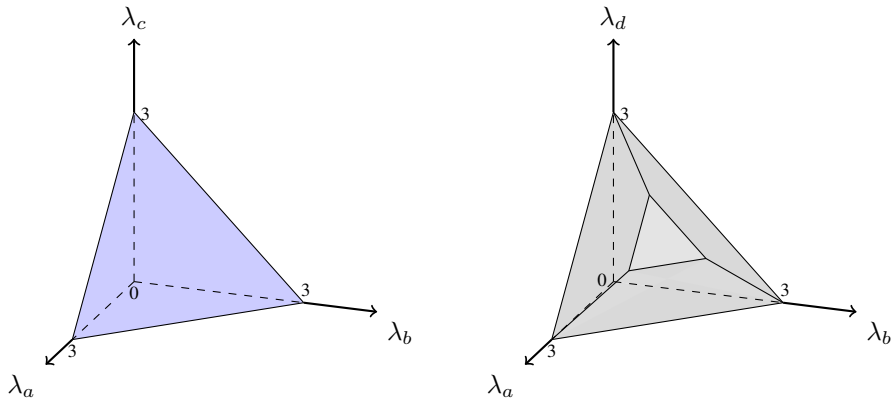

This example demonstrates that the maximum service rate for a data symbol can differ based on the choice of the generator matrix, given that not all other data symbols have zero demand (service rate). Although the volume of the service rate region remains constant, the orientation of the convex service polytope may vary depending on the selected generator matrix. When the demand for a particular data symbol exceeds that of others, we select the generator matrix accordingly for the same code.
For instance, let $a$, $b$, $c$, and $d$ represent files stored in a coded storage system that employs a binary Hamming coding scheme. If the demand for file $d$ is greater than the demand for the other files, then the generator matrix presented in Example \ref{ex1} is appropriate, as $(1,1,1,2) \in \Lambda(G)$.

In contrast, if more users are interested in accessing file $a$ than the other files, a more suitable generator matrix would be 

\[
G' = \begin{pmatrix}
1 & 1 & 1 & 1 & 0 & 0 & 0 \\
1 & 0 & 1 & 0 & 1 & 0 & 0 \\
0 & 1 & 1 & 0 & 0 & 1 & 0 \\
1 & 1 & 0 & 0 & 0 & 0 & 1 
\end{pmatrix}.
\]

This matrix allows for the recovery of file $a$ from the recovery set consisting of non-systematic nodes only, yielding $(2, 1, 1, 1) \in \Lambda(G')$. 

Consider a situation where the cumulative demand for objects $a$, $b$, and $c$ is high. In this case, the preferred generator matrix would be $G'$ as we have  $\lambda_a + \lambda_b + \lambda_c \leq 4$. However, as for generator matrix $G$, we have $\lambda_a + \lambda_b + \lambda_c \leq 3.$

\begin{remark}
    Let $\mathcal{I} \subseteq [k]$. Based on the demand for the data symbols indexed by elements of the set $\mathcal{I}$, we can pick the generator matrix that maximizes the cumulative service rate by finalizing the associated parity-check matrix. The alternative approach is to rearrange the order of the data symbols in the message vector to selectively enhance the overall service rate for the set $\mathcal{I}$.
\end{remark}

 We can determine the number of equations that limit the sum of service rates for a set of three data objects to $3$. This is useful for approximating the SRR when the demand for data files changes continuously, while the demand for any three specific files remains non-zero; that is, when $|\mathcal{I}| = 3$.

\vspace{0.2cm}
\noindent \textbf{Counting the number of inequalitites of the form $\mathbf{ \lambda_i+\lambda_j+\lambda_k\leq 3}$}

 Consider $v_i$ to be a column of the parity-check matrix corresponding to the systematic server of the $i$-th data symbol in $G$. In other words, $v_i\in \mathbb{F}_2^r$, such that $\text{wt}(v_i)\geq 2$.
 There are $\binom{2^r-1-r}{2}$ inequalities of the form $\lambda_i+\lambda_j\leq 3$ for systematic Ham($r, 2$). However, not all sums $v_i+v_j$ correspond to a data symbol; some may relate to a parity-check. 
 
 Our goal is to determine the number of combinations such that $v_i+v_j=v_s$, where each column corresponds to distinct data symbols at coordinate places $i, j, s\in[k]$ respectively. We will denote this number by $M_3$.

We first determine the number of combinations such that $v_i+v_j=e$, where $e$ is a standard basis vector in $\mathbb{F}_2^r$. Suppose $v_i=\{x_1, x_2, \ldots, x_r\}$ with Supp$(v_i)= \{i_1, i_2, \ldots, i_t\}$, where $2\leq t\leq r$. For $v_i+v_j=e$, we require \[\text{wt}(v_j)=
    t-1\text{ or }t+1,
\]

as the vectors $v_i$ and $v_j$ differ in exactly one position. 
\begin{itemize}
    \item There are $t$ vectors with weight $t-1$ and support $\{i_1, i_2, \ldots, i_t\}\setminus\{i_w\}$, where $w\in[t]$ respectively.
    \item There are $r-t$ vectors with weight $t+1$ and support $\{i_1, i_2, \ldots, i_t\}\cup\{i_u\}$, where $u\in [r]\setminus\text{Supp}(v_i).$ 
   \end{itemize}Note that for $t = 2$, there are no vectors $v_j$ of weight $t - 1$ that correspond to a systematic server in $G$. Therefore, we only consider vectors with weight $t + 1$. Similarly, any combination that involves vectors of weight $t$ and $t-1$ has already been accounted for when $t$ was set to $(t-1)$. Consequently, we conclude that there are $r-t$ vectors $v_j$ for each vector $v_i$ of weight $t$ such that $v_i+v_j=e$. 
   
   In $\mathbb{F}_2^r$, the total number of vectors of weight $t$ is given by $\binom{r}{t}$. This leads us to conclude that the total number of combinations of the form $v_i + v_j = e$ is equal to 

\[
\sum_{t=2}^{r-1} \binom{r}{t}(r - t).
\]
Therefore, the total number of combinations of the type $v_i + v_j = v_s$ is
\begin{align*}
    M_3&=\frac{\binom{2^r-1-r}{2}-\sum_{t=2}^{r-1}\binom{r}{t}(r-t)}{3},\\
&=\frac{\binom{2^r-1-r}{2}-r2^{r-1}+r^2}{3}.
\end{align*}

\section{Conclusion}\label{iv}
In this paper, we studied the recovery $G$ system of systematic Ham($r, q$), where $G$ is a generator matrix of Ham($r, q$). Through optimal request splitting, it was demonstrated that the size of the maximal achievable simplex in SRR of systematic Ham($r, q$) is $\frac{2q-1}{q-1}$, which is independent of $r$. By considering the recovery sets as hyperedges and using elementary properties of the dual code, we identified limits on the sum of service rates of data symbols contained in a given subset of all data symbols of systematic binary Hamming codes. We showed that in systematic Ham($r, 2$), the cumulative service rate cannot surpass the total number of data symbols $2^r-1-r$ for $r>3$, and for any two data symbols, $i$-th and $j$-th file, the combined service rate $\lambda_i+\lambda_j$ is less than $3$. For future work, further analysis of the hypergraph and projective space context could lead to a complete characterization of the SRR of Hamming codes. This concept may be extended to codes with higher locality, such as locally recoverable codes and maximally recoverable codes. Investigating the link between the SRR of the code and the dual code’s structure also presents an interesting research direction.

\section*{Acknowledgments}
The authors express their gratitude to the University Grants Commission (UGC), Government of India, for the financial support received under the UGC NET-JRF scheme.

\ifCLASSOPTIONcaptionsoff
  \newpage
\fi



%

%

\vfill


\end{document}